\documentclass[runningheads]{llncs}
\usepackage[T1]{fontenc}
\usepackage{graphicx}
\usepackage{dependencies}
\usepackage{cite}
\usepackage{braket}
\usepackage{amsmath,amssymb,amsfonts}
\usepackage{algorithmic}
\usepackage{graphicx}
\usepackage{textcomp}
\usepackage{xcolor}

\usepackage{hyperref}
\begin{document}
\title{Quantum Interior Point Methods: A Review of Developments and An Optimally Scaling Framework}
\titlerunning{Quantum Interior Point Methods}
\author{Mohammadhossein Mohammadisiahroudi\inst{1,2,3}\orcidID{000-0002-4046-0672} \and
Zeguan Wu\inst{1,4}\orcidID{0000-0002-5695-7579} \and
Pouya Sampourmahani\inst{1}\orcidID{0000-0002-2292-551X}\and
Adrian Harkness\inst{1}\orcidID{0009-0001-5518-6442}\and
Tam\'as Terlaky\inst{1}\orcidID{0000-0003-1953-1971}}
\authorrunning{Mohammadisiahroudi et al.}
\institute{Industrial and Systems Engineering, Lehigh University, Bethlehem, PA 18015, USA
\and
Mathematics and Statistics, University of Maryland, Baltimore County, MD 21250, USA 
\and
Quantum Science Institute, University of Maryland, Baltimore County, MD 21250, USA 
\and 
Computer Science, University of Pittsburgh, Pittsburgh, PA 15260, USA}
\maketitle             
\begin{abstract}
The growing demand for solving large-scale, data-intensive linear and conic optimization problems, particularly in applications such as artificial intelligence and machine learning, has highlighted the limitations of classical interior point methods (IPMs). Despite their favorable polynomial-time convergence, conventional IPMs often suffer from high per-iteration computational costs, especially for dense problem instances. Recent advances in quantum computing, particularly quantum linear system solvers, offer promising avenues to accelerate the most computationally intensive steps of IPMs. However, practical challenges such as quantum error, hardware noise, and sensitivity to poorly conditioned systems remain significant obstacles. In response, a series of Quantum IPMs (QIPMs) has been developed to address these challenges, incorporating techniques such as feasibility maintenance, iterative refinement, and preconditioning. In this work, we review this line of research with a focus on our recent contributions, including an almost-exact QIPM framework. This hybrid quantum-classical approach constructs and solves the Newton system entirely on a quantum computer, while performing solution updates classically. Crucially, all matrix-vector operations are executed on quantum hardware, enabling the method to achieve an optimal worst-case scalability w.r.t dimension, surpassing the scalability of existing classical and quantum IPMs.

\end{abstract}
\section{Introduction}
In this paper, we review recent advances in Quantum Interior Point Methods (QIPMs) for linear optimization (LO) problems. The standard form LO problem is  minimizing a linear objective function over a polyhedron, formally defined as
\begin{equation} \label{eq:primall problem}\tag{P}
    \begin{aligned}
    \min_{x\in \mathbb{R}^{n}}\  c^T&x \\
    \text{s.t. }
    Ax &= b, \\
    x &\geq 0,
    \end{aligned}
\end{equation}
where $A\in \mathbb{R}^{m\times n}$ with rank$(A)=m$, $b\in\mathbb{R}^m$, and $c\in\mathbb{R}^n$. It is well-known that there is a dual problem associated with the primal problem, as
\begin{equation} \label{eq:dual problem}\tag{D}
    \begin{aligned}
    \max_{(y,s)\in \mathbb{R}^{m}\times\mathbb{R}^{n}} \  b^Ty\ \ & \\
    \text{s.t. }
    A^Ty +&s = c,\\
    &s \geq 0.
    \end{aligned}
\end{equation}
By the strong duality theorem \cite{roos2005interior}, all optimal solutions, if they exist, belong to the set $\mathcal{PD}^*$, which is defined as
\begin{align*}
\mathcal{PD}^*=&\{(x,y,s)\in\mathbb{R}^{n+m+n}:\ Ax=b,\ A^Ty+s=c,\\ & x^Ts=0, \ (x,s)\geq0 \}.
\end{align*} 

Linear optimization plays a foundational role in a broad range of fields, including machine learning, operations research, logistics, and finance. Historically, the Simplex algorithm \cite{bertsimas1997introduction} was among the first prominent methods to solve LO problems. While highly effective in many practical instances, Simplex methods can exhibit exponential-time behavior in the worst case \cite{klee1972good}. In contrast, the introduction of Interior Point Methods (IPMs) marked a major breakthrough in optimization. Starting with Karmarkar's projective algorithm \cite{Karmarkar1984_New}, IPMs have evolved into the most theoretically efficient class of algorithms for solving LO problems, offering polynomial-time complexity with robust convergence guarantees.
Using the standard form of LO problems, the set of feasible primal-dual solutions is defined as
\begin{equation*}
\mathcal{PD}=\left\{(x,y,s)\in\mathbb{R}^{n} \times \mathbb{R}^m\times\mathbb{R}^n |\ Ax=b,\ A^Ty+s=c, \ (x,s)\geq0\right\}.
\end{equation*}
Then, the set of all feasible interior solutions is
\begin{equation*}
\mathcal{PD}^0=\left\{(x,y,s)\in\mathcal{PD}\ |\ (x,s)>0\right\}.
\end{equation*}
By the Strong Duality theorem, all optimal solutions, if there exist any, belong to the set $\mathcal{PD}^*$ defined as
\begin{equation*}
\mathcal{PD}^*=\left\{(x,y,s)\in\mathcal{PD} \ | \ x^Ts=0\right\}.
\end{equation*}
Let $\zeta \geq 0$, then the set of $\zeta$-optimal solutions can be defined as
\begin{equation*}
\mathcal{PD}_{\zeta} = \left\{(x,y,s)\in\mathcal{PD} \ \Big|\ \frac{x^Ts}{n}\le \zeta \right\}.
\end{equation*}

Modern IPMs exploit the geometry of the central path, an analytic trajectory defined by a set of perturbed optimality conditions, which guides iterates toward the optimal solution \cite{roos2005interior, nesterov1994interior}. When initialized appropriately, IPMs follow this path using Newton's method, requiring at most $\mathcal{O}(\sqrt{n}\log(1/\epsilon))$ iterations to obtain an $\epsilon$-approximate solution \cite{roos2005interior}. However, a significant computational bottleneck in IPMs lies in solving the Newton linear system at each iteration. Classical direct methods such as Cholesky factorization incur $\mathcal{O}(n^3)$ complexity, which becomes intractable for large-scale, dense problems. Iterative methods, including conjugate gradient (CG) solvers \cite{al2009convergence, Monteiro2003_Convergence}, mitigate this challenge with lower per-iteration costs but at the expense of increased sensitivity to matrix conditioning and convergence accuracy.

To further enhance the scalability of IPMs, several improvements have been introduced. These include partial update techniques and low-rank updates, which reduce the cost of computing Newton directions and yield the best-known classical total complexity of $\mathcal{O}(n^3 L)$ for LO problems \cite{roos2005interior}. More recently, the incorporation of advanced tools such as fast matrix multiplication, spectral sparsification, and stochastic methods have pushed the complexity to $\mathcal{O}(n^{\omega}\log(n/\epsilon))$, where $\omega < 2.3729$ is the matrix multiplication exponent \cite{brand2020, cohen2021solving, lee2015efficient}. Alternatively, first-order methods like the primal-dual hybrid gradient (PDHG) algorithm have demonstrated empirical success in solving large-scale LO problems, although they lack rigorous theoretical complexity bounds \cite{chambolle2011first, applegate2021practical}.

Alongside these classical advances, quantum computing has emerged as a powerful paradigm capable of accelerating various computational tasks. Quantum algorithms such as Shor's factoring algorithm \cite{shor1994algorithms} and Grover's search algorithm \cite{grover1996fast} have showcased the potential of quantum computers to achieve polynomial or even exponential speedups. Of particular interest for optimization is the class of quantum linear system algorithms (QLSAs), pioneered by the Harrow-Hassidim-Lloyd (HHL) algorithm \cite{harrow2009quantum}. HHL and its successors \cite{childs2017quantum, ambainis2012variable, wossnig2018Quantum} solve sparse quantum linear systems with exponential speedups under certain assumptions, although they exhibit unfavorable dependence on condition number, sparsity, and required precision.

Motivated by the capabilities of quantum computing, researchers have sought to integrate quantum solvers into classical optimization frameworks. This effort has led to the development of Quantum Interior Point Methods (QIPMs), which aim to exploit quantum acceleration in solving the Newton systems arising in IPMs. Notable contributions include quantum subroutines for the Simplex method \cite{nannicini2024fast}, QAOA for binary optimization \cite{farhi2014quantum}, and quantum multiplicative weight update methods for semidefinite optimization \cite{brandao2019quantum, van2018improvements}. For linear and semidefinite programming, QIPMs have demonstrated potential polynomial speedups in terms of problem dimension \cite{kerenidis2021quantum, augustino2023quantum}. However, these early QIPMs faced substantial challenges. The hybrid nature of QIPMs necessitates the extraction of classical information from quantum states at each iteration, typically via quantum tomography algorithms (QTAs). These steps often introduce significant errors and computational overhead, diminishing the overall efficiency of the method.

To overcome these limitations, a series of research efforts has led to the development of improved QIPMs. By incorporating iterative refinement and preconditioning techniques, recent frameworks reduce the impact of quantum errors and ill-conditioning, achieving exponential improvements with respect to precision and condition number compared to earlier quantum methods \cite{mohammadisiahroudi2024efficient, mohammadisiahroudi2025improvements, wu2023inexact,mohammadisiahroudi2025quantum}. For instance, Wu et al. \cite{wu2024quantum} introduced a dual logarithmic barrier-based QIPM with improved iteration complexity and memory access efficiency via QRAM.

In this work, we review some of these advancements with a focus on the novel, almost-exact QIPM framework that achieves provable quantum advantage. In our proposed approach, the Newton system is both constructed and solved entirely on a quantum computer, while classical computation is reserved only for solution updates. All matrix-vector products, the most expensive components in classical QIPMs, are offloaded to quantum hardware, reducing total runtime. Our hybrid quantum-classical framework achieves optimal worst-case scaling of $\mathcal{O}(n^2)$ for fully dense linear optimization problems, outperforming both classical IPMs and existing QIPMs in terms of dimensional complexity.

This framework supports inexact quantum operations, such as quantum matrix inversion and matrix-vector/matrix-matrix multiplication, through the use of iterative refinement. Crucially, unlike prior QIPMs, our method eliminates all classical matrix operations, resulting in a total classical arithmetic cost of $\mathcal{O}(n^2 \log(1/\epsilon))$. This asymptotically improves upon previous QIPMs by a factor of $\mathcal{O}(\sqrt{n})$ and offers a provable quantum speedup, as any classical analog would require at least $\mathcal{O}(n^{2.5})$ operations.

The structure of the paper is as follows. In Section~\ref{sec: feasible}, we discuss how novel reformulations aid in maintaining feasibility and achieving the best-known iteration complexity for QIPMs. Section~\ref{sec: IRPRE} reviews how iterative refinement and preconditioning techniques can mitigate the effects of ill-conditioning and enhance the precision of QIPMs. In Section~\ref{sec: QLSA}, we review recent advancements in Quantum Linear System Algorithms (QLSAs) and Quantum Tomography. Section~\ref{sec: ALmost} presents the state-of-the-art QIPM based on a novel Almost-Exact IPM framework that achieves optimal scaling. In Section~\ref{sec: app}, we explore the applications and implications of recent QIPM advancements in artificial intelligence and machine learning. Finally, Section~\ref{sec: con} concludes the paper and outlines directions for future work.

\section{Inexact Feasible QIPMs}\label{sec: feasible}
In the general scheme of IPMs, we apply the Newton method to the perturbed optimality conditions iteratively to approach an optimal solution by tracing the so-called central path. There are three reformulations of Newton systems to calculate the Newton direction at each step of IPMs in the classical IPM literature. The prevailing system is the Normal Equation System (NES), defined as
\[
A D^2 A^T \Delta y = A x - \beta \mu A S^{-1} e,
\]
where $\beta<1$, $A \in \mathbb{R}^{m \times n}$ is the constraint matrix, $S=\text{diag}(s)$, $$D = \text{diag}(x)^{1/2} \text{diag}(s)^{-1/2}$$ is the diagonal scaling matrix, and $\mu = \frac{x^T s}{n}$ is the complementarity measure. 

One major issue is that an inexact solution to any traditional Newton system calculated by a QLSA+QTA subroutine may lead to infeasibility. To properly address this infeasibility, inexact infeasible QIPM (II-QIPM) \cite{mohammadisiahroudi2024efficient} has been developed, which has $\Ocal(n^2\log(\frac{1}{\epsilon}))$ iteration complexity, where $n$ is the number of variables and $\epsilon$ is the target optimality gap. 

To improve this complexity, we propose two inexact feasible QIPMs (IF-QIPMs) using two novel reformulations of Newton systems. First, we use a basis for the null-space of $A$, stored in columns of the matrix $V$, to reformulate the Newton system in the Orthogonal Subspaces system (OSS) \cite{mohammadisiahroudi2023inexact} as
\begin{equation}\label{eq: OSS}\tag{OSS}
  \begin{bmatrix}
-XA^T&SV
\end{bmatrix}\begin{bmatrix}
\Delta y\\
\lambda
\end{bmatrix}=\beta \mu e- Xs.
\end{equation}
We prove that the inexact solution for the OSS system provides a feasible Newton direction.

In another paper, we propose another system that is a modified version of the NES, and it is more adaptable for quantum singular value transformation \cite{mohammadisiahroudi2025improvements}. We prove the iteration complexity for both IF-QIPMs is $\Ocal(\sqrt{n}\log(\frac{1}{\epsilon}))$, which leads to considerable polynomial speed-up in the complexity of QIPMs. 

\section{Iterative Refinement and Preconditioning}\label{sec: IRPRE}
Another challenge in QIPMs is that their complexity has polynomial dependence on $\frac{1}{\epsilon}$ because of the QTA's overhead. This means previous QIPMs are not polynomial time algorithms as one needs to reach $\frac{1}{\epsilon}=\Ocal(2^{L})$ to find an exact optimal solution for an LO problem, where $L$ is the binary length of input data. We use an iterative refinement technique to address this issue \cite{mohammadisiahroudi2025improvements,mohammadisiahroudi2023inexact}.

Iterative refinement has been widely used in classical numerical algorithms to improve the accuracy of solutions to linear systems. We adapt this technique to iteratively use limited-precision IF-QIPM to obtain a higher-precision solution. We prove that iteratively refined IF-QIPMs (IR-IF-QIPMs) have exponentially improved complexity with respect to precision compared to previous QIPMs.

The last challenge in QIPMs is that QLSAs are sensitive to the condition number of linear systems arising in QIPMs, and Newton systems are usually ill-conditioned. There are two major sources of ill-conditioning in QIPMs. First, for degenerate LO problems, the sequence of coefficient matrices of Newton systems converges to a singular matrix. i.e., their condition number grows to infinity. We show that a properly adapted iterative refinement technique helps with issues, as we stop QIPMs early when the condition number is comparatively small enough. Another source of ill-conditioning is the ill-conditioned input matrix $A$. We address this issue by preconditioning the Newton system. In addition, we show how this particular preconditioner can be applied on a quantum machine without excessive cost \cite{mohammadisiahroudi2025improvements}.

\section{Improved QLSA+QTA Subroutine for QIPMs}\label{sec: QLSA}
The idea of using iterative refinement for quantum algorithms is further used to develop an improved QLSA+QTA subroutine for QIPMs \cite{mohammadisiahroudi2024quantum}.
The most efficient QLSA to solve a linear system of the form $Mz = \sigma$ with $\widetilde{\Ocal}_{\frac{p}{\epsilon}}(\kappa\|M\|_F)$ queries\footnote{The $\widetilde{\Ocal}_{\alpha, \beta} \left( g(x) \right)$ notation indicates that quantities polylogarithmic in $\alpha, \beta$ and $g(x)$ are suppressed.} to QRAM \cite{chakraborty2018power}, where $p$ is the system dimension, representing an exponential speedup over classical algorithms.
A major hurdle lies in the fact that QLSAs solve Quantum Linear System Problems (QLSPs) and so the result is a quantum state, which deviates from the classical definition of the solution of LSPs. Consequently, a Quantum Tomography Algorithm (QTA) is essential to extract a classical solution. The best time complexity of QTA is $\Ocal(\frac{p\varrho}{\epsilon})$, where $\varrho$ represents the upper bound on the norm of the solution.

The overall complexity of QLSA and QTA combined is $\widetilde{\Ocal}_{\frac{p}{\epsilon}}(\frac{p\kappa^2\|\sigma\|}{\epsilon})$. In comparison to the conjugate gradient method (CGM), its query complexity exhibits a better dependence on sparsity, with unfavorable dependence on precision and condition number. An iterative classical-quantum linear system algorithm (ICQLSA) has been proposed, which exponentially improves the time complexity of Quantum Linear Solvers, providing a classical solution with high precision up to $\widetilde{\Ocal}_{\frac{p|\sigma|}{\epsilon}}(p\kappa^2)$ queries to QRAM \cite{mohammadisiahroudi2024quantum}.

This new advancement enables us to do the calculations in high precision settings where $\epsilon=2^{-2L}$, which is almost exact for the solution of the LO problems. Thus, we can also do matrix-vector multiplications on the quantum machine. Using ICQLSA and Quantum mat-vec product within the IR-IF-QIPM leads to optimal scaling $\Ocal(n^2\kappa_A L)$ as the worst case complexity for solving LO can not have better dimension dependence than quadratic, as storing and reading dense matrix $A$ needs $\Ocal(n^2)$ arithmetic operations.

\section{The state-of-the-art QIPM}\label{sec: ALmost}
In this section, we develop an almost exact quantum interior point method for solving linear optimization problems. Assuming that the input data is all integers, we denote the binary length of the input data by 
\begin{align*}
    L&=mn+m+n+\sum_{i,j}\lceil\log_2(|a_{ij}|+1)\rceil\\
&+\sum_{i}\lceil\log_2(|c_{i}|+1)\rceil+\sum_{j}\lceil\log_2(|b_{j}|+1)\rceil,
\end{align*}
where $a_{ij}$ represents the $ij$-element of matrix $A$. The optimal partition is also defined as 
\begin{align*}
    \Bcal&=\{j\in\{1,\dots,n\}:x^*_j>0 \text{ for some }(x^*,y^*, s^*)\in \mathcal{PD}^*\},\\
    \Ncal&=\{j\in\{1,\dots,n\}:s^*_j>0 \text{ for some }(x^*,y^*, s^*)\in \mathcal{PD}^*\}.
\end{align*}
The following lemma is a classical result first proved by \cite{Khachiyan1980}.
\begin{lemma}\label{lemma: L bound}
Let $(x^*,y^*,s^*)\in \mathcal{PD}^*$ be a basic solution. If $x_i^*>0$, then we have $x_i^*\geq 2^{-L}$. If $s_i^*>0$, then we have $s_i^*\geq 2^{-L}$.
\end{lemma}

Lemma~\ref{lemma: L bound} is a fundamental result in the complexity analysis of IPMs. It means that after enough number of iterations of IPMs, a decision variable can be rounded to zero if it is smaller than $2^{-L}$. Then, by a rounding procedure, one can find an exact optimal solution for linear optimization \cite{roos2005interior, wright1997primal}. In the proposed algorithm, all calculations happen on a quantum machine with precision $\epsilon=2^{-tL}$ where $t$ is a small constant, less than $10$. This high level of accuracy justifies describing the algorithm as \emph{almost-exact}. The only calculation that happens on a classical computer is updating the solution and vector-vector summation. In this paper, we use the dual logarithmic barrier method, which has a simple framework. At each step of the dual log barrier IPM, we need to solve the following Newton system
\begin{equation}
\begin{bmatrix}
    I & A^T \\
    AS^{-2}  & 0
\end{bmatrix} \begin{bmatrix}
    \Delta s \\
    \Delta y
\end{bmatrix} =
\begin{bmatrix}
    0 \\
    \frac{1}{\mu}(b-AS^{-1}e)
\end{bmatrix},
\end{equation}
where $S={\rm diag}(s)$.
Let $\hat{\Delta s}= S^{-2} \Delta s$, we can have the system 
\begin{equation}\label{eq: AS}
\begin{bmatrix}
    S^{2} & A^T \\
    A  & 0
\end{bmatrix} \begin{bmatrix}
    \hat{\Delta s} \\
    \Delta y
\end{bmatrix} =
\begin{bmatrix}
    0 \\
    \frac{1}{\mu}(b-AS^{-1}e)
\end{bmatrix}.
\end{equation}
One can easily verify that $M=\begin{bmatrix}
    S^{2} & A^T \\
    A  & 0
\end{bmatrix}$ is a symmetric positive definite matrix, and so system \eqref{eq: AS} has a unique solution \cite{roos2005interior}. Given $s$, one can build block-encodings of implementing  matrix $M$ and preparing state $\sigma =\begin{bmatrix}
    0 \\
    \frac{1}{\mu}(b-AS^{-1}e)
\end{bmatrix} $ efficiently, assuming that matrix $A$ stored in QRAM in advance. The general steps of the proposed almost exact QIPM using a short-step framework are described in Algorithm~\ref{alg:AE-QIPM}.

The dual logarithmic barrier method starts with a strictly feasible dual solution $(y^0,\ s^0)\in \mathcal{D}^\circ$ and a $\mu^0>0$ such that $(y^0,\ s^0)$ is close to the $\mu^0$-center in the sense of the proximity measure $\delta(s^0,\mu^0)$, which, \textcolor{blue}{for given $s$ and $\mu$,} is defined as,
\begin{equation*}
    \begin{aligned}
        \delta(s, \mu) := \left\| s^{-1}\Delta s \right\|_{2}.
    \end{aligned}
\end{equation*}

\begin{algorithm}
\caption{Almost Exact QIPM (AE-QIPM)}\label{alg:AE-QIPM}
    \begin{algorithmic}
        \STATE \textbf{INPUT} Dual feasible solution $(y^0,s^0)$, $\mu^0>0$, $0<\theta<1$, and $\delta \left((y^0, s^0), \mu^0\right) < \frac{1}{2}$, where $\delta$ is the proximity measure from \cite{roos2005interior,wu2024quantum}
        \STATE Store $A, b,c $ on QRAM
        \STATE $k \gets 1$
        \WHILE{$\mu > 2^{-2L}$}
        \STATE $(\Delta y^{k}, \hat{\Delta s}^{k} ) \gets $ Solve system \ref{eq: AS} with precision $\epsilon = 2^{-tL}$
        \STATE $y^{k+1}\gets y^{k}+ \Delta y^{k}$
        \STATE $s^{k+1}\gets s^{k}+ (S^{k})^{2} \hat{\Delta s}^{k}$
        \STATE $\mu^{k+1} = (1-\theta) \mu^k$
        \STATE $k \gets k+1$
        \ENDWHILE
    \end{algorithmic}
\end{algorithm}
As we analyze the worst-case complexity, we assume $m=\Ocal(n)$ and matrices are fully dense.
\begin{theorem}\label{theo: IPM}
    The number of iterations in Algorithm~\ref{alg:AE-QIPM} is at most $$\Ocal (\sqrt{n}L).$$
\end{theorem}

\subsubsection{Proof of Theorem~\ref{theo: IPM}}
We start with a strictly feasible solution $(x^0,y^0,s^0)$. In the dual logarithmic barrier IPM, we do not compute the value of $x$ and $y$, but they exist. We have
\begin{equation*}
    \begin{aligned}
        Ax^0 = b,\ A^T y^0 + s^0 = c, \ s^0>0.
    \end{aligned}
\end{equation*}
Then we use a quantum subroutine to compute an inexact $\Delta s^0$ with associated error $\xi^1$. After a full Newton step, we have
\begin{equation*}
    Ax^1 = b,\ A^T y^1 + s^1 = c + \xi^1,\ s^1>0.
\end{equation*}
Now we get a feasible solution for problem~$(A,b,c+\xi_1)$. We do another full Newton step, then we have
\begin{equation*}
    \begin{aligned}
        Ax^2 = b,\ A^Ty^2 + s^2 = c+\xi^1 + \xi^2,\ s^2>0.
    \end{aligned}
\end{equation*}
We can keep doing this until we have
\begin{equation*}
    \begin{aligned}
        Ax^k = b,\ A^Ty^k + s^k = c+ \sum_{i=1}^k\xi^i,\ s^k>0.
    \end{aligned}
\end{equation*}
Then, we can rewrite all of them into
\begin{equation*}
    \begin{aligned}
        Ax^k = b,\ A^Ty^k + s^k +\sum_{j=1}^k\xi^{j} = c + r^k,\ s^k>0,
    \end{aligned}
\end{equation*}
where $r^k = \sum_{j=1}^k\xi_j$.
This implies we obtained a series of feasible iterates for problem $(A,b,c+r^k)$. 
The Newton steps are inexact but satisfy the conditions in \cite{wu2024quantum}, the $\Ocal(\sqrt{n}L)$ complexity still holds.
However, these Newton steps are artificial steps because we do not know exactly the errors $\xi^i$. We need to show that the actual Newton steps we inexactly compute are close enough to these artificial Newton steps, and the inexactness is acceptable for the convergence conditions.

In the first iteration, the actual and artificial Newton steps are computed as
\begin{equation*}
    \begin{aligned}
        \Delta s^0 &= -A^T \left( A (S^0)^{-2} A^T \right)^{-1} \frac{1}{\mu^0} \left( b- \mu^0 A (S^0)^{-1} e \right) +\xi^1\\
        \Delta \tilde{s}^0 &= -A^T \left( A (\tilde{S}^0)^{-2} A^T \right)^{-1} \frac{1}{\mu^0} \left( b- \mu^0 A (\tilde{S}^0)^{-1} e \right),
    \end{aligned}
\end{equation*}
where
\begin{equation*}
    \begin{aligned}
        \tilde{S}^0 = s^0+ r^k.
    \end{aligned}
\end{equation*}
According to \cite{wu2024quantum}, we need 
\begin{equation*}
    \begin{aligned}
        \left\|  (\tilde{S}^0)^{-1} (\Delta s^0 - \Delta \tilde{S}^0)\right\|_2 \leq 0.1\delta_{\tilde{c}}(\tilde{s}^0, \mu^0),
    \end{aligned}
\end{equation*}
where $\delta_{\tilde{c}}$ is the proximity measure for the perturbed problem $(A, b, \tilde{c})$ with $\tilde{c} = c + r^k$.
This condition can be guaranteed when
\begin{equation}\label{eq: condition iter1}
    \begin{aligned}
        \left\| (S^0 (\tilde{S}^0)^{-1} (I - S^0 (\tilde{S}^0)^{-1})) \right\|_2 &\leq 0.033 \delta_{\tilde{c}}(\tilde{s}^0, \mu^0),\\
        \left\| I - (S^0 (\tilde{S}^0)^{-1})^2 \right\|_2 &\leq 0.033,\\
        \left\| (\tilde{S}^0)^{-1} \xi^1 \right\|_2 &\leq 0.033 \delta_{\tilde{c}}(\tilde{s}^0, \mu^0).
    \end{aligned}
\end{equation}
Notice that all three conditions can be satisfied by pushing $\xi^i\leq 2^{-tL}$. 
This proves that our inexact Newton step is a feasible inexact Newton step for the perturbed problem. 
Then, according to Theorem 3.3 of \cite{wu2024quantum}, we have the $\mathcal{O}(\sqrt{n}L)$ complexity. Additionally, it is easy to verify that $t=4$ satisfies all requirements.

After $\Ocal(\sqrt{n}L)$ iterations, we have an $\tilde{x}>0$ such that 
 \begin{align*}
        A\tilde{x}&=b, \\
        A^T y^k + s^k &= c+r^k,\\
        (\tilde{x})^T s^k &\leq 2^{-2L},
    \end{align*}
where $r^k=\sum_{i=1}^k\xi^i$. It is easy to verify that $(\tilde{x}, y^k, s^k)$ is a $2^{-2L}$-optimal solution for the perturbed problem $(A,b,c+r^k)$, and one can calculate the exact optimal solution by a rounding procedure. It is easy to verify that $\|r^k\|\leq 2^{(1-t)L}$. In the remaining, we show how we can retrieve an optimal solution of the original problem with a rounding procedure from the optimal solution for the perturbed problem.

It is straightforward to see that $(\tilde{x}, y^k, s^k)$ is in a $2^{(1-t)L}$-neighborhood of the optimal set for the original problem $(A,b,c)$. As the smallest nonzero element of $s^*$ and $x^*$ is greater than $2^{-L}$, using partitions $B$ and $N$ of this solution, by solving a constrained least squares problem, an optimal solution for the original problem can be obtained. For the details of the rounding procedures, refer to Chapter 7 of \cite{wright1997primal}. 

It is worth noting that the rounding procedures are strongly polynomial-time methods. They can also be quantized using quantum linear system solvers; however, we do not explore the cost and implementation details of rounding procedures in this paper, as it is beyond the scope of this paper.

\subsection{Quantum Subroutine}\label{sec: Sub}
In this section, we analyze the complexity of building and solving system \eqref{eq: AS}. We use the general scheme of the Quantum Tomography framework of \cite{mohammadisiahroudi2024quantum,mohammadisiaroudi2023exponentially}. We assume that we have access to a large enough QRAM, and we store data $A, b, c$ initially on QRAM with worst-case $\Ocal(n^2)$ complexity. At each state, we need to store $s$ on QRAM and build and solve System \eqref{eq: AS} using the iterative quantum linear solver of \cite{mohammadisiaroudi2023exponentially}.
\begin{algorithm}[H]
\caption{Quantum Linear Solver}\label{alg:QLSA}
    \begin{algorithmic}
        \STATE \textbf{INPUT} $(A,b,c)$ stored on QRAM,
        \STATE Store $s$ on QRAM
        \STATE $k \gets 1$
        \STATE $z^k \gets 0$
        \WHILE{$\|\Delta y^k - \Delta y^{k-1}\|  > 2^{-4L}$}
        \STATE Prepare State $\ket{r^k} = \ket{\sigma - M z^k}$
        \STATE Apply inverse of block encoding of $M$ using QSVT \cite{chakraborty2018power}
        \STATE Extract classical solution $\frac{p^k}{\|p^k\| }= \frac{M^{-1}r^k}{\|M^{-1}r^k\|} $ via Tomography \cite{van2023quantum} with precision $\epsilon = 10^{-2}$
       \STATE Estimate norm of $\|p^k\|$ and $\|r^k\|$
        \STATE $z^{k+1}\gets z^{k}+ \frac{p^{k}}{\|r^k\|}$
        \STATE $k \gets k+1$
        \ENDWHILE
    \end{algorithmic}
\end{algorithm}
At each iteration of Algorithm~\ref{alg:QLSA}, the only classical operation is updating the solution by a vector summation with $\Ocal(n)$ arithmetic operations. In the following, we calculate the cost of quantum operations.
\begin{lemma}\label{lem: block1}
    Given $A$ and $S$ stored on QRAM, the following statements are true:
    \begin{itemize}
        \item We can construct a block-encoding of $M$ using $\Ocal(\text{polylog}(\frac{n}{\epsilon}))$ queries to QRAM.
        \item We can prepare the the state $\ket{r}$ using $\Ocal(\text{polylog}(\frac{n}{\epsilon}))$ queries to QRAM.
        \item We can apply $M^{-1}$ using $\tilde{\Ocal}_{n,\kappa, \frac{1}{\epsilon}}(\kappa\|A\|_F)$ queries to QRAM. 
        \item Norm estimation of $p^k$ and $r^k$ costs $\tilde{\Ocal}_{n,\kappa, \frac{1}{\epsilon}}(\kappa\|A\|_F)$ queries to QRAM.
    \end{itemize}
\end{lemma}
The proof of Lemma \ref{lem: block1} is the direct result of Prepositions 1 to 6 of \cite{augustino2023quantum}. 

\begin{lemma}\label{lem: qta iteration}
    The number of iterations of Algorithm~\ref{alg:QLSA} is at most $\Ocal(L)$.
\end{lemma}
The proof of Lemma \ref{lem: qta iteration} is based on \cite{mohammadisiaroudi2023exponentially}. Additionally, the total complexity of Algorithm~\ref{alg:QLSA} is based on the analysis provided by \cite{mohammadisiaroudi2023exponentially}. 

\begin{theorem}
    Assuming $(A, b, c)$ is stored on QRAM, Algorithm~\ref{alg:QLSA} can find a $2^{-tL}$-precise solution for System~\eqref{eq: AS} with 
    $$\tilde{\Ocal}_{n\kappa L}(n\kappa\|A\|_F)$$
    iterations.
\end{theorem}

\subsection{Proposed IR-AE-QIPM}\label{sec: IR}
In this section, we discuss how to use the iterative refinement method (IR) for LO problems to improve complexity as in \cite{wu2024quantum}, and provide the full description of our proposed algorithm. The first IR method for LO has been proposed by \cite{gleixner2016iterative}. Mohammadisiahroudi et al. \cite{mohammadisiahroudi2024efficient} first showed that using iterative refinement can improve the complexity of QIPMs w.r.t precision and condition number. Further, in \cite{mohammadisiahroudi2023inexact}, a quadratically convergent iterative refinement scheme was proposed for feasible IPMs. An IR method for dual log-barrier QIPM has been developed in \cite{wu2024quantum}.

In \cite{wu2024quantum}, the iterative refinement method for the LO problem works as follows:
\begin{itemize}
    \item[1.] Start with the original problem and solve it to a low accuracy.
    \item[2.] If the accuracy of the original problem is not enough, construct a refining problem using the current iteration values; otherwise, the algorithm halts.
    \item[3.] Solve the refining problem to a low accuracy and update the solution to the original problem; then, go to step 2.
\end{itemize}

In our proposed algorithm, after each solve, we have a feasible solution to a perturbed problem. To use the iterative refinement method, we need to construct a solution to the original problem from the solution to a perturbed problem.
To do so, we need a projection procedure. We use Algorithm~\ref{alg:QLSA} to solve the following problem
\begin{equation*}
    \begin{aligned}
        \min_{y} \|A^T y + s_k - c\|_2,
    \end{aligned}
\end{equation*}
which is equivalent to
\begin{equation*}
    \begin{aligned}
        AA^T y = A(c-s_k).
    \end{aligned}
\end{equation*}
Then we have 
\begin{equation*}
    \begin{aligned}
        s = c - A^T y.
    \end{aligned}
\end{equation*}
According to the argument in the previous section, this $(y, s)$ is feasible for the original problem with a duality gap bounded by twice the low accuracy.
Then, we can use the IR scheme to refine the solution to high accuracy as in \cite{wu2024quantum}.

To get the full complexity of the proposed algorithm, we discuss the accuracy needed for $\xi^i$. In the first iteration, we need conditions \eqref{eq: condition iter1}. Theoretically, $\Delta_{\tilde{c}}$ might be zero, which implies the corresponding Newton system right-hand side is zero. We do not need to solve such Newton systems. Instead, we check the norm of the right-hand side vector. If the norm is too small $(\leq 2^{-4L})$, we update $\mu$ without computing the Newton step. Then, conditions \eqref{eq: condition iter1} can be guaranteed when
\begin{equation*}
    \begin{aligned}
        \|\xi^i\|_2 \leq {\rm poly}\left(\frac{2^{-4L} }{n\kappa_{A(s^0)^{-1}}} \right) \approx {\rm poly} (2^{-4L}), \ \forall i\in[k].
    \end{aligned}
\end{equation*}
This bound also works for the remaining iterations.
Now, we present the pseudocode of our proposed algorithm and the main theorem.

\begin{algorithm}
\caption{Iteratively Refined Almost Exact QIPM (IR-AE-QIPM)}\label{alg:IR-AE-QIPM}
    \begin{algorithmic}
        \STATE \textbf{INPUT} Dual feasible solution $(y^0,s^0)$, $\mu^0>0$, $0<\theta<1$, $\delta \left((y^0, s^0), \mu^0\right) < \frac{1}{2}$, $\nabla^{(0)} = 1$, $0<\zeta\ll\tilde{\zeta}$
        \STATE Store $A, b,c $ on QRAM
        \STATE $k \gets 1$
        \STATE $(y_1, s_1) \gets$ Solve dual problem with accuracy $\tilde{\zeta}$
        \WHILE{$\nabla^{(k-1)} < \frac{1}{\zeta}$}
        \STATE $\nabla^{(k)} \gets \nabla^{(k-1)}\times\frac{1}{\tilde{\zeta}}$
        \STATE Construct the IR problem as in \cite{wu2024quantum}
        \STATE $(\hat{y}, \hat{s}) \gets $ Solve IR problem with accuracy $\tilde{\zeta}$ and project into proper subspace
        \STATE $y^{k+1}\gets y^{k}+ \frac{1}{\nabla^{(k)}} \hat{y}$
        \STATE $s^{k+1} \gets c- A^T y^{(k)}$
        \STATE $k \gets k+1$
        \ENDWHILE
    \end{algorithmic}
\end{algorithm}

\begin{theorem}[Lemma 13 of \cite{wu2024quantum}]\label{theo: IR iter}
    Algorithm \ref{alg:IR-AE-QIPM} terminates after $\Ocal\left(\frac{\log (\zeta)}{\log (\hat{\zeta})}\right)$ iterations.
\end{theorem}
For our purpose, we use $\zeta=2^{-tL}$ and $\hat{\zeta}$ is constant. Thus, the outer loop has $\Ocal(L)$ iteration bound. We also have a matrix-vector product at each step of this IR scheme with cost $\Ocal(n^2)$ classic arithmetic operations.

The major challenge in AE-QIPM Algorithm~\ref{alg:AE-QIPM} is that the complexity of the quantum solver depends on the condition number, and the condition number grows in each iteration of AE-QIPM. As in IR Algorithm~\ref{alg:IR-AE-QIPM}, we stop AE-QIPM early at fixed precision. It has been shown that with early termination $\kappa^{(k)}=\Ocal(\kappa_0)$ where $\kappa_0$ is the condition number of the coefficient matrix for $(y^0, s^0)$ and it is constant \cite{wu2024quantum}.

\subsection{Total Complexity}\label{sec: total}
In this section, we put together all the elements discussed in the previous sections to calculate the total worst-case complexity of IR-AE-QIPM.
\begin{theorem}\label{theorem: main}
    Algorithm~\ref{alg:IR-AE-QIPM} produces a $2^{(1-t)L}$ precise optimal solution of the LO problem using at most $$\tilde{\Ocal}_{\kappa_0, n, \|A\|_F}(n^{1.5}L\kappa_0))$$ queries to QRAM and $\mathcal{O}(n^2L)$ classical arithmetic operations.
\end{theorem}

\begin{proof}
    The number of IR iterations is bounded by $\Ocal(L)$ based on Theorem~\ref{theo: IR iter}. At each iteration, we have $\Ocal(n^2)$ classical arithmetic operations due to a classical matrix-vector product and the cost of AE-QIPM to solve the refining problem. Additionally, to address $\|A\|_F$ in the complexity, one can initially normalize data by $\|A\|_F$, and consequently, the final precision should be increased by $\|A\|_F$, which appears in polylog. The quantum complexity is $\tilde{\Ocal}_{n,L, \|A\|_F}(n^{1.5}\kappa_0L)$ queries to QRAM, and $\Ocal(nL)$ arithmetic operations at each step of AE-QIPM. Thus, the total number of queries to QRAM is 
    $$\tilde{\Ocal}_{\kappa_0, n}(n^{1.5}L\kappa_0)),$$ 
    and the total number of classical arithmetic operations is bounded by 
    $\mathcal{O}(n^2L).$
\end{proof}

Table~\ref{tab:compelxities} compares the complexity of the proposed IR-AE-QIPM with other classical and quantum IPMs. As we can see, the total complexity of our approaches outperforms previous complexities. In the last line of the table, we show the complexity of the classical counterpart of the IR-AE-IPM using CG to solve the system. As we can see, the total complexity can not be better than $n^{2.5}$ in the classical version. This exhibits a clear quantum advantage compared to other algorithms in the literature. It should be mentioned that the quantum complexity of all QIPMs is the number of queries to QRAM. Without QRAM assumptions, some overhead cost may appear in complexities, although the quantum central path method of \cite{augustino2023central} is QRAM-free.
\renewcommand{\arraystretch}{1.5}
\begin{table*}[ht]
\caption{Worst-case Complexity of different IPMs for LO $(m=n)$}
\label{tab:compelxities}
\centering
\resizebox{\textwidth}{!}{%
\begin{tabular}{cccc}
\hline
Algorithm                & Linear System Solver & Quantum Complexity & Classical Complexity  \\ \hline\hline
IPM with Partial Updates \cite{roos2005interior}            &            Low rank updates          &                    &        $\Ocal(n^{3}L)$                       \\ \hline
Feasible IPM \cite{roos2005interior}            &  Cholesky             &                    &                $\Ocal(n^{3.5}L)$               \\ \hline
II-IPM   \cite{Monteiro2003_Convergence}                          & PCGM                   &                    &                $\Ocal(n^{5}L\bar{\chi}^2)$        \\ \hline
Robust IPM \cite{brand2020}                        & Fast Mat-Mul and Partial Update                  &                    &                $\Ocal(n^{w}L)$              \\
\hline\hline
Quantum Central Path \cite{augustino2023central}
 &          Hamiltonian Evolution      &          $\tilde{\Ocal}(n^{3.5}\frac{\omega}{\epsilon})$          &                              \\ \hline

IR-IF-IPM \cite{mohammadisiahroudi2025improvements}                      & PCGM                  &                    &    $\tilde{\Ocal}_{\mu^0}(n^{3.5}L\bar{\chi}^2)$     \\ \hline
IR-IF-QIPM \cite{mohammadisiahroudi2023inexact}                      & QLSA+QTA             &             $\tilde{\Ocal}_{n,\kappa_{A}, \|A\|,\|b\|,\mu^0}(n^{1.5}L\kappa_{A}^2\omega^5)$       &          $\tilde{\Ocal}_{\mu^0}(n^{2.5}L)$                  \\ \hline
 IR-IF-QIPM \cite{mohammadisiahroudi2025improvements}                      & Precond+QLSA+QTA             &      $ \widetilde{\Ocal}_{ n, \left\| A \right\|_F, \frac{1}{\epsilon}} ( n^{1.5} L\bar{\chi}^2      ) $              &            $\tilde{\Ocal}_{\mu^0}(n^{2.5}L)$                 \\ \hline
 IPM with approximate Newton steps~\cite{apers2023quantum}     & Q-spectral Approx.     & $\tilde{\mathcal{O}}_{n, \frac{1}{\zeta}}(n^{5.5})$                   & $\tilde{\mathcal{O}}_{\frac{1}{\zeta}}(n^{1.5} )$  \\ \hline
Quantum Dual-log Barrier \cite{wu2024quantum}                                       & QLSA+QTA               & $\widetilde{\mathcal{O}}_{n, \kappa_0, \mu^0,\|A\|_F}\left(n^{1.5} \kappa_0 L \right)$ & $\mathcal{O}( {n}^{2.5}L)$ \\ \hline
 Proposed IR-AE-QIPM                       & IQLSA+Quant Mat-Vec           &             $\tilde{\Ocal}_{n,\kappa_{0},\|A\|_F}(n^{1.5}L\kappa_{0})$       &          ${\Ocal}(n^{2}L)$                    \\ \hline
 Classical IR-AE-IPM                       & CGM             &                 &            ${\Ocal}(n^{2.5}L\kappa_0)$                  \\ \hline
\end{tabular}%
}
\vspace{3pt}

\begin{minipage}{\textwidth}
{\footnotesize \textit{Note.} Quantum complexity is expressed in the form of query complexity. Here, $\omega$ is the upper bound on the norm of the optimal solution, and $\bar{\chi}^2$ is an upper bound on the condition number of the preconditioned NES.  $w$ is the matrix multiplication parameter. PCGM stands for Preconditioned Conjugate Gradient Method.}
\end{minipage}
\end{table*}
\section{Applications in AI and Machine Learning}\label{sec: app}
The integration of QLSAs and QIPMs has shown the potential to accelerate core problems in machine learning. This section highlights some key applications that can benefit from these quantum techniques, as demonstrated in recent studies including \cite{wu2023inexact,mohammadisiahroudi2022regression}.

Quantum-enhanced regression is one of the most direct applications of QLSAs in machine learning. Ordinary Least Squares (OLS), Weighted Least Squares (WLS), and Generalized Least Squares (GLS) problems can all be reduced to solving linear systems of the form $(X^TX)\beta = X^Ty$, which QLSAs can handle efficiently. When paired with quantum tomography algorithms (QTAs), these solvers can retrieve classical solutions for model training and inference. The incorporation of iterative refinement techniques further enables exponential speedups with respect to precision, overcoming the classical bottleneck caused by ill-conditioning. Specifically, QLSAs offer:
\begin{itemize}
    \item Exponential speedup w.r.t. dimension in state preparation.
    \item Exponential speedup w.r.t. precision via iterative refinement.
    \item Milder dependence on condition number through adaptive regularization.
    
\end{itemize}

Many sophisticated machine learning models, such as Support Vector Machines (SVMs) and Lasso Regression, can be formulated as Linearly Constrained Quadratic Optimization (LCQO) problems. These problems are ideal candidates for QIPMs, which leverage QLSAs to solve Newton systems arising in Interior Point Methods. Wu et al. \cite{wu2023inexact} proposed an Inexact Feasible QIPM (IF-QIPM) that preserves the feasibility of iterates using orthogonal subspace systems (OSS), enabling the solution of LCQO problems, including:
\begin{itemize}
    \item Lasso Regression: Promotes sparse solutions using $\ell_1$ regularization. Reformulated as an LCQO problem, it can be solved with improved complexity using IF-QIPMs.
    \item Soft-Margin Support Vector Machines: Reformulated as LCQO using variable splitting and slack variables. QIPMs achieve better complexity in high-dimensional regimes.
\end{itemize}

These quantum-enhanced formulations demonstrate:
\begin{itemize}
    \item Polynomial speedup w.r.t. dimension $n$ over classical IPMs.
    \item Exponential speedup w.r.t. precision and condition number over previous QIPMs.
    \item Improved feasibility guarantees through OSS-based feasibility maintenance.
\end{itemize}

Together, these applications showcase the growing relevance of QLSAs and QIPMs in machine learning, particularly as hardware capabilities advance. The combination of quantum speedups in dimension, precision, and matrix conditioning illustrates a compelling path forward for quantum-enhanced data science.
\section{Conclusions}\label{sec: con}

In this paper, we presented recent advances in the development of Quantum Interior Point Methods (QIPMs) for Linear Optimization. By integrating iterative refinement and preconditioning techniques, we tackled two major challenges inherent in QLSA-based QIPMs: the inexactness of quantum solvers and their sensitivity to the condition number of the Newton system. We further introduced a novel Almost-Exact Interior Point Method framework, in which all matrix-vector operations and Newton system computations are performed on a quantum computer. This approach delivers a provable quantum speedup over classical IPMs.

To achieve an exponentially small error in the computed Newton steps, we embed iterative refinement both internally within the quantum solver and externally across IPM iterations. As a result, the overall algorithm attains an optimal worst-case complexity of $\mathcal{O}(n^2 L)$ for fully dense linear optimization problems.

A key limitation of the proposed method is its dependence on Quantum Random Access Memory (QRAM), the physical implementation of which remains an open challenge. However, alternative approaches can be explored to mitigate this dependency. For example, circuit-based QRAM constructions \cite{park2019circuit}, recent developments in Quantum Singular Value Transformation (QSVT) without block encoding \cite{chakraborty2025quantum}, or sparse-access input models offer promising directions for developing QRAM-free variants of the algorithm. Moreover, a detailed resource estimation study, such as the framework in \cite{tu2025towards}, is essential for evaluating the real-world feasibility and quantum advantage of the proposed method.

Future research could also focus on extending this framework to a primal-dual Almost-Exact QIPM applicable to both linear and semidefinite optimization problems. Primal-dual methods, particularly those based on self-dual embedding formulations, offer the advantage of not requiring an initial strictly feasible interior point, thus expanding the applicability of QIPMs in practice.
\begin{credits}
\subsubsection{\ackname} This work is supported by Defense Advanced Research Projects Agency as part of the project W911NF2010022: {\em The Quantum Computing Revolution and Optimization: Challenges and Opportunities}.

\subsubsection{\discintname}
The authors have no competing interests to declare that are
relevant to the content of this article.
\end{credits}

{\hyphenpenalty=100000
\bibliographystyle{IEEEtranS}
\bibliography{references}}
\end{document}